\documentclass{article}
\usepackage[utf8]{inputenc}
\usepackage{amsthm, amssymb, mathrsfs, amscd, amsmath, stmaryrd}
\usepackage{mathtools}
\usepackage{tikz}
\usepackage{hyperref}
\usepackage{amsfonts}
\usepackage{enumerate}
\usepackage{cite} %BibTeX
\usepackage{soul}
\usepackage[n, operators, sets, adversary, landau, probability, notions, logic, ff, mm, primitives, events, complexity, oracles, asymptotics, keys]{cryptocode}

\providecommand{\verifyx}{\pcalgostyle{Verify}}

\newcommand{\Hpub}{H_{\text{pub}}}

\theoremstyle{plain}
\newtheorem{Th}{Theorem}
\newtheorem{Prop}[Th]{Proposition}

\theoremstyle{definition}

\title{Security issues of CFS-like digital signature algorithms}
\author{Giuseppe D'Alconzo, Alessio Meneghetti, Paolo Piasenti}
\begin{document}

\maketitle

\begin{abstract}
    We analyse the security of some variants of the CFS code-based digital signature scheme. We show how the adoption of some code-based hash-functions to improve the efficiency of CFS leads to the ability of an attacker to produce a forgery compatible to the rightful user's public key.
\end{abstract}

\section{Introduction}

With the discovery and the increasingly closer advent of quantum computers, the most adopted signature schemes (e.g. DSA \cite{kerry2013fips}, ECDSA \cite{ecdsa}, EdDSA \cite{eddsa}, Schnorr \cite{schnorr}) are often considered not secure anymore because they are well known to be broken by Shor's algorithm \cite{shor1994algorithms}. 
Possible countermeasures are obtained by the exploitation of schemes whose security relies on NP-hard problems, or, more in general, on problems whose solutions are thought to be difficult in both the classic and quantum frameworks of computation. Among these alternatives, some of the most prominent are represented by lattice-based cryptography, multivariate polynomial cryptography, hash-based cryptography and interactive identification schemes. These branches of post-quantum cryptography are all present among the finalists of the NIST Post-Quantum Standardization process\footnote{NIST Post-Quantum Standardization process webpage: \url{https://csrc.nist.gov/Projects/post-quantum-cryptography/post-quantum-cryptography-standardization}, Accessed: 2021-12-01} (the interested reader can see the overview \cite{DBLP:journals/corr/abs-2107-11082}). Notably, code-based digital signature algorithms are not in this list. It is worth mentioning that Classic McEliece \cite{chou2020classic} is a code-based Key-Encapsulation Mechanism among the finalists of the competition regarding post-quantum key-agreement protocols, and about forty years of cryptanalysis have shown its resiliency and security (Classic McEliece is based on the works of McEliece \cite{mceliece1978public} and Niederreiter \cite{niederreiter1986knapsack}). Instead, the initial code-based digital signature schemes did not pass the first round of selection. 
Although since the earliest and historical proposal CFS \cite{courtois2001achieve} (which will be discussed later on) plenty of ideas and projects have followed, the issue of finding a viable candidate is still an open and tough problem. The key points here are two: of course security, but also efficiency. The main drawback of CFS is the signing time, in fact the message is hashed with a counter until the digest is a decodable syndrome.

Further signature schemes have been provided with KKS \cite{kabatianskii1997digital} and its variants \cite{barreto2011one, gaborit2012efficient}. This scheme converts the message in a decodable syndrome, using a different strategy with respect to CFS, but, taking into account the attack highlighted in \cite{cayrel2007kabatianskii}, all the variants have to be considered at best one-time signature schemes; additionally, strong caution has to be taken in the choice of parameters, as shown by \cite{otmani2011efficient} which broke all the parameters proposed in \cite{kabatianskii1997digital,kabatiansky2005error,barreto2011one}.

On the other hand, most CFS-like schemes persist to be unbroken, despite their slow speed in the signing process due to the attempt-based design. 
Some schemes try to reduce the signing time using the idea behind KKS: instead of searching for a decodable syndrome through the hash of the message, a map that aims to the space of decodable syndromes is used. An example is the mCFS$_c$ signature \cite{ren2017efficient}, that hashes the message into a decodable syndrome using a code-based hash function. Unfortunately, in this work we prove that this approach does not work, leaving room to an attacker to forge a valid signature without knowing the secret key.

This work is organized as follows: in the first section we present the notation and some basic notions from Coding Theory, then we introduce the two signature schemes CFS and mCFS. The second section presents the variant mCFS$_c$ and the concerning code-based hash function. We show an attack on this construction. The third and final section generalizes the strategy adopted in the mCFS$_c$ signature and shows that such approach leads to an attack.

\subsection{Notation}
With $\mathbb{F}_2$ we denote the field with $2$ elements and with $(\mathbb{F}_2)^n$ the vector space of dimension $n$ over $\mathbb{F}_2$. An $[n,k]$ \emph{binary code} $\mathcal{C}$ is a vector subspace of $(\mathbb{F}_2)^n$ of dimension $k$. The elements of $\mathcal{C}$ are called \emph{codewords}. Every $[n,k]$ binary code can be represented as the kernel of a $(n-k)\times n$ matrix $H$ called \emph{parity-check matrix}. The syndrome of a vector $v\in(\mathbb{F}_2)^n$ is given by $s=Hv^{\top}$ and the \emph{Hamming weight} of a vector is the number of its non-zero coordinates. With $\parallel$ we denote the concatenation of strings or vectors.

\subsection{Digital signatures and CFS}
A digital signature is a public-key cryptosystem consisting in three algorithms: $\kgen$ is the key generation algorithm that takes in input a security parameter $\lambda$ and outputs the pair of secret and public keys $(\sk, \pk)$, a signature algorithm $\sign$ that on input $\sk$ and a message $m$ outputs a signature $\sigma$ of $m$, and a verifier algorithm $\verifyx$ in which, given a public key $\pk$, a message $m$ and a signature $\sigma$, it verifies if the signature of the message is valid and is generated by the corresponding secret key $\sk$.\\
A digital signature algorithm must satisfy some security proprieties: authentication, non repudiation, integrity, non reusability and unforgeability. See \cite{kerry2013digital} for a more detailed study on digital signatures.

The CFS algorithm \cite{courtois2001achieve} consists in producing a signature exploiting the Niederreiter public-key cryptosystem \cite{niederreiter1986knapsack}. This scenario entails a substantial difference with respect to RSA, for instance: since trapdoor functions allow digital signatures taking advantage of the unique capability of public key owner to invert those functions, it is clear that only the messages whose hashes fall within the ciphertext space can be signed in this way. In our framework, we would like to deal with a linear code for which there exists an efficient decoding algorithm and for which the set of decodable syndromes (namely the ciphertext space) is as big as possible. In formal terms, given a $(n-k) \times n$ parity-check matrix $H$ of such a code, this translates into having a meaningful portion of vectors $s \in \mathbb{F}_{2}^{n-k}$ for which there exists a corresponding error pattern $e \in \mathbb{F}_{2}^{n}$ of Hamming weight less than the correcting capability of the code $t$, such that the syndrome of $e$ is $s$.
Since the fact that the union of the spheres centered in codewords and of radius $t$ covers the whole space $\mathbb{F}_{2}^{n}$ only happens in the case of perfect codes (which are banally unusable because of the overmuch leak of information they would disclose) the smartest play to make remains to repeatedly hash the message until one obtains a decodable syndrome. Binary Goppa codes \cite{berlekamp1973goppa} represent the best choice as underlying code for both their efficient decoding method (Patterson algorithm) and their steady resistance against all known attacks. This procedure is nothing more than a “hash-and-sign” routine, which inevitably requires several tries. Concretely, given suitable hash function $h$, one produces a sequence $d_0,\ldots, d_\iota$ of elements in $\mathbb{F}_{2}^{n-k}$ such that
$$d_0 = h(m\parallel 0),\; d_1 = h(m\parallel 1),\; \ldots, \; d_i = h(m\parallel i),\;\ldots\; d_\iota = h(m\parallel\iota) $$
where $\iota$ is the smallest integer such that $d_\iota$ is a decodable syndrome. The signature is then composed by the corresponding error pattern $e_\iota$ (that only the signer can compute) and by the counter $\iota$. The first straightforward question that arises is about how many attempts are needed in order to obtain a useful syndrome. The answer can be easily found by comparing the total number of syndromes to the number of (efficiently) correctable syndromes:
$$ \frac{\textsc{\# decodable syndromes}}{\textsc{\# total syndromes}} = \frac{\sum_{i = 0}^{t}\binom{n}{i}}{2^{n-k}} \simeq \frac{\binom{n}{t}}{2^{n-k}} \simeq \frac{\frac{n^{t}}{t!}}{n^t}=\frac{1}{t!} $$
which represent the probability of a syndrome to be decodable (here the relations among the parameters of a generic binary Goppa code have been used, i.e. $k = n - mt$ and $n=2^m$).\\
This scheme bases its security on two assumptions: the hardness of both the Syndrome Decoding Problem \cite{berlekamp1978inherent} and the Goppa Code Distinguisher Problem \cite{courtois2001achieve}.\\
Now we present the signature scheme of CFS.
\begin{itemize}
    \item $\kgen_{\text{CFS}}(1^\lambda)$: select $n,k,t$ according to the security parameter $\lambda$ then pick a random $[n,k]$ binary Goppa code $\mathcal{C}$ with correcting capacity $t$ and parity-check matrix $H$ and let $\mathcal{D}_H$ be an (efficient) syndrome decoding algorithm for $\mathcal{C}$. Pick a random $(n-k)\times(n-k)$ invertible matrix $S$ and a random $n\times n$ permutation matrix $P$ and set $\Hpub=SHP$. Chose a hash function $h$. Output $\pk=(h, t,\Hpub)$ as public key and $\sk=(S,H,P,\mathcal{D}_H)$ as secret key.
    
    \item $\sign_{\text{CFS}}(m, \sk)$: given the message $m$, compute $d_i=h(m\parallel i)$, starting from $i=0$ and increase it until $d_{\iota}$ is a decodable syndrome. Set $e=\mathcal{D}_H(S^{-1}d_{\iota})$ and output the signature $\sigma=(\iota,eP)$.
    
    \item $\verifyx_{\text{CFS}}(m, \sigma, \pk)$: let $\sigma=(\iota,u)$, verify that $u$ has Hamming weight less or equal than $t$, then compute $a = h(m||\iota)$ and $b = \Hpub u^{\top}$. The signature $\sigma$ is valid if and only if $a = b$.
\end{itemize}
We can see that the signature is correct, in fact
$$ a = h(m\parallel \iota ) = d_{\iota} = SH \cdot e^{\top} = SHPP^{-1}\cdot e^{\top} = \Hpub\cdot u^{\top} = b$$
where $S^{-1}d_{\iota} = H \cdot e^{\top}$ comes from the fact that $e$ has syndrome $S^{-1}d_{\iota}$.

In \cite{dallot2007towards} authors propose a modified version of the CFS signature called mCFS. Here the counter $i$ used in $\sign$ is replaced by a random nonce.

\section{The mCFS$_c$ signature}

In this section we describe and analyze the signature in \cite{ren2017efficient}, and in Proposition \ref{prop: attack mcfsc} we explicit an attack.

\subsection{Code Based Hash Function}
This signature is build on the protocol mCFS \cite{courtois2001achieve,dallot2007towards} using a particular code based hash function. This hash function is based on the work of \cite{augot2005family} and it follows the Merkel-Damgard design \cite{damgaard1989design}. Let $r$ be the length of the digest and let $s$ be an integer. The hash function is the iterative application of a \emph{compression function} $f:(\mathbb{F}_2)^s\to(\mathbb{F}_2)^r$, in fact, given a string $m$ proceed as following:
\begin{enumerate}
    \item consider $m$ padded such that its length is a multiple of $s$ and split $m$ in $ \lvert m\rvert/s$ blocks of length $s$;
    \item in the first round, combine the first block of $m$ with a fixed initial vector (IV) obtaining the state $L_1$ of length $s$ and compute $f(L_1)$;
    \item in the $i$-th round combine $f(L_{i-1})$ with the $i$-th block of $m$ obtaining the $i$-th state $L_i$ and apply $f$ to it;
    \item the output of the hash function is given by $f(L_{\lvert m \rvert/s})$.
\end{enumerate}
The hash function used in mCFS$_c$ uses the scheme above and the following compression function $f$. Let $r$ be the length of the digest. Given a $r\times n$ parity-check matrix $H$ of a $[n,n-r]$ binary code $\mathcal{C}$, let $w$ be an integer dividing $n$ and set $l=n/w$ and $s=w\log(l)$. Now we describe the compression function $f:(\mathbb{F}_2)^{s} \to (\mathbb{F}_2)^{r}$ based on $H$:
\begin{enumerate}
    \item let $H_1,\dots,H_{w}$ be $r\times w$ matrices such that $H = (H_1, \dots, H_{w})$;
    \item given $x\in(\mathbb{F}_2)^{s}$, split it in $w$ blocks of length $\log(l)$: $x=(x_1,\dots,x_n)$. We can see each $x_i$ as an integer between $0$ and $l-1$;
    \item set $f(x)$ as the sum of the $(x_i+1)$-th column of the matrix $H_i$, for $i=1,\dots,w$. In formulas, if $(H_i)_j$ is the $j$-th column of $H_i$, we have
    $$f(x) = \sum_{i=1}^w (H_i)_{x_i+1}.$$
\end{enumerate}
Observe that $f$ strongly depends on the choice of the parity-check matrix $H$.

For the signature mCFS$_c$, in \cite{ren2017efficient}, $H$ is chosen as the parity-check matrix of a $[n,n-r]$ binary Goppa code, and the parameter $w$ is less than the correcting capacity $t$ of the code. This yields to the hash function $h_H:\{0,1\}^* \to (\mathbb{F}_2)^{r}$ based on $H$. Observe that the computation of $h_H$ implies the knowledge of $H$.
\begin{Prop}\label{prop:hash}
    For every state $L_i$ of the hash function $h_H$, $f(L_i)$ is a syndrome of a vector of Hamming weight $w$.
\end{Prop}
\begin{proof}
    By construction, the state $x=L_i$ is splitted in $w$ integers $x_1,\dots,x_w$ between 0 and $l-1$. Let $c_i$ be the vector of length $n$ having support $(x_1+1)+0\cdot l, (x_2+1)+1\cdot l,\dots, (x_w+1)+(w-1)l$, it has Hamming weight $w$ and $f(L_i)$ is exactly $Hc_i^{\top}$, the syndrome of $c_i$.
\end{proof}

We can summarize the compression function as follows: let $n$ and $w$ be positive integers such that $w$ divides $n$ and set $s=w\log(n/w)$. Consider the bijection
\begin{equation*}
\begin{split}
    \mathrm{split}: (\mathbb{F}_2)^s & \to ((\mathbb{F}_2)^{\log(n/w)})^w\\
    (u_1,\dots,u_s) & \mapsto (z_1,\dots,z_w)
\end{split}
\end{equation*}
that splits a binary vector of length $s$ into $w$ vectors of length $\log(n/w)$. Now we can see every vector in $(\mathbb{F}_2)^{\log(n/w)}$ as an integer between $0$ and $n/w-1$. Define
\begin{equation}\label{eq_deltat}
\begin{split}
    \delta_t:(\mathbb{F}_2)^s & \to (\mathbb{F}_2)^n\\
    (u_1,\dots,u_s) & \mapsto (v_1,\dots,v_n)
\end{split}
\end{equation}
where $(v_1,\dots,v_n)$ is the vector of Hamming weight $w$ whose support is given by $(\mathrm{split}(x)_1+1)+0\cdot \frac{n}{w},(\mathrm{split}(x)_2+1)+1\cdot \frac{n}{w},\dots, (\mathrm{split}(x)_w+1)+(w-1)\frac{n}{w}$. Then the compression function $f$ can be written as $f(x)=H\cdot\delta_t(x)$.

\subsection{The signature scheme}
Since in \cite{ren2017efficient} it is not specified if the hash function is based on the secret matrix $H$ or the public matrix $\Hpub$, we first observe that since the hash function is part of the public key, this discloses the secret matrix $H$ and an attack can be performed confronting columns of $H$ and $\Hpub$, finding the permutation in quadratic time. Hence, assuming that the hash is based on the public matrix, the signature is given by the following algorithms. Let $\lambda$ be the security parameter.
\begin{itemize}
    \item $\kgen_{\text{mCFS}_c}(1^\lambda)$: select $n,k,t$ according to $\lambda$ then pick a random $[n,k]$ binary Goppa code $\mathcal{C}$ with correcting capacity $t$ and parity-check matrix $H$ and let $\mathcal{D}_H$ be an (efficient) syndrome decoding algorithm for $\mathcal{C}$. Pick a random $n\times n$ permutation matrix $P$ and set $\Hpub=HP$. Choose an integer $w$ less than $t$ and such that $w$ divides $n$ and construct the hash function $h_{\Hpub}:\{0,1\}^* \to (\mathbb{F}_2)^{n-k}$ based on $\Hpub$. Output $\pk=(h_{\Hpub},t,\Hpub)$ as public key and $\sk=(H,P,\mathcal{D}_H)$ as secret key.
    
    \item $\sign_{\text{mCFS}_c}(m, \sk)$: given the message $m$, pick a random $R$ in $\{1,\dots,2^{n-k}\}$ and compute $d=h_{\Hpub}(h_{\Hpub}(m)\parallel R)$ and set $e=\mathcal{D}_H(d)$. Output the signature $\sigma=(R,eP)$.
    
    \item $\verifyx_{\text{mCFS}_c}(m, \sigma, \pk)$: let $\sigma=(R,u)$. Verify that $u$ has Hamming weight less or equal than $t$, then compute $a = h_{\Hpub}(h_{\Hpub}(m)||R)$ and $b = \Hpub u^{\top}$. The signature $\sigma$ is if and only if $a = b$.
\end{itemize}
The signature scheme is correct using the same argument for CFS.

\begin{Prop}\label{prop: attack mcfsc}
Let $(\sk,\pk)$ be the output of $\kgen_{\text{mCFS}_c}(1^\lambda)$. An attacker knowing the public key $\pk$ can forge a signature compatible to the private key $\sk$ for any message $m$.
\end{Prop}
\begin{proof}
Given a public key $\pk=(h_{\Hpub},t,\Hpub)$ and a message $m$, the attacker picks a random $R$ in $\{1,\dots,2^{n-k}\}$ and computes $d=h_{\Hpub}(h_{\Hpub}(m)\parallel R)$ but it stops before the last round of the outer hash function $h_{\Hpub}$, obtaining the round state $\bar{L}=L_{\lvert m \rvert/s}\in (\mathbb{F}_2)^s$ instead of the digest $d=f(L_{\lvert m \rvert/s})$. He set $u=\delta_t(\bar{L})$ and outputs as a signature for $m$ the tuple $\sigma=(R,u)$. Anyone can verify that this is a valid signature of $m$ compatible with the secret key $\sk=(H,P,\mathcal{D}_H)$, in fact we can compute $a = h_{\Hpub}(h_{\Hpub}(m)||R)$ and $b = \Hpub u^{\top}$ and observing that $a$ is equal to $b$ since multiplying $u=\delta_t(\bar{L})$ by $\Hpub$ is the last step of the hash function $h_{\Hpub}$. Therefore $\sigma$ is a valid signature.
\end{proof}

\subsection{A generalisation of mCFS$_c$}
We now slightly generalise mCFS$_c$ by considering a modification of $h_H$, proving that this new entire family of hash functions is vulnerable to the same attack we described for $h_H$ and therefore is not suitable for secure applications.
\\
Let $B_{n,t}$ be the set of vectors in $(\mathbb{F}_2)^n$ of Hamming weight less or equal than $t$.
Let $\gamma_t:\left(\mathbb{F}_2\right)^s\to\left(\mathbb{F}_2\right)^n$ be such that $\mathrm{Im}(\gamma_t)\subseteq B_{n,t}$, i.e. $\gamma_t$ is a function mapping bitstrings of length $s$ into bitstring of length $n$ with a Hamming weight bounded by $t$:
$$
\mathrm{w}(\gamma_t(v))\leq t\quad \forall v\in\left(\mathbb{F}_2\right)^s
$$
We denote with $\bar{h}_H:\{0,1\}^*\to\left(\mathbb{F}_2\right)^{n-k}$ the function mapping messages into syndromes associated to the parity-check matrix $H$ defined by the formula
\begin{equation}\label{eq_barh}
m\mapsto \bar{h}_{H}(m)=H\cdot \gamma_t(h(m))\;,
\end{equation}
where $h(\cdot)$ is any efficient function $\{0,1\}^*\to\left(\mathbb{F}_2\right)^{s}$. For simplicity of notation, we will call $h$ a hash function, even though we do not require here that $h$ satisfy any security property (even though it would be a good practice to choose a cryptographically-secure hash).

With this definition we can consider the following version of CFS, that we call $\widetilde{\text{CFS}}$:
\begin{itemize}
    \item $\kgen_{\widetilde{\text{CFS}}}(1^\lambda)$: randomly choose a code $\mathcal{C}$ to be used in the CFS algorithm (i.e. $\mathcal{C}$ is a code for which there exists an efficient decoder up to $t$ errors and whose randomly picked equivalent codes are indistinguishable from random) with parity-check matrix $H$ and efficient syndrome decoding algorithm $\mathcal{D}_H$. Then randomly choose an invertible $(n-k)\times (n-k)$ matrix $S$ and a permutation $n\times n$ matrix $P$, and define $\Hpub=SHP$. Choose an efficient map $\gamma_t$ and a hash $h$. Output $\sk=(S,H,P,\mathcal{D}_H)$ as the secret key and $\pk=(h,\gamma_t,t,\Hpub)$ as the public key.
    
    \item $\sign_{\widetilde{\text{CFS}}}(m, \sk)$: given the message $m$, compute $d=\bar{h}_{\Hpub}(m)$ according to \eqref{eq_barh}. Decode $S^{-1}d$ with the decoder for $H$ and thus obtaining an error vector $e=\mathcal{D}_H(S^{-1}d)$ of Hamming weight at most $t$ and then compute $\bar{e}=eP$. Output the signature $\sigma=\bar{e}$.
    
    \item $\verifyx_{\widetilde{\text{CFS}}}(m, \sigma, \pk)$: verify that $\sigma=\bar{e}$ has Hamming weight less or equal than $t$, then compute $a=\bar{h}_{\Hpub}(m)$ and $\bar{b}=\Hpub\bar{e}$. The signature is valid if $a=b$.
\end{itemize}
The correctness of the above signature scheme is straightforward and follows directly from the correctness of CFS.

We remark that mCFS$_c$ is (basically) obtained by adopting the algorithm above where:
\begin{itemize}
\item[-] $\mathcal{C}$ is a binary irreducible Goppa code;
\item[-] $S=I_{n-k}$ is the identity matrix of order $n-k$;
\item[-] $\gamma_t$ is the map $\delta_t$ defined in \eqref{eq_deltat};
\item[-] $h$ is the code-based hash function $h_{\Hpub}$ stopped before the last application of $\delta_t$ and multiplication by $H_{\mathrm{pub}}$, which we denote momentarily $h_{\Hpub}^{\mathrm{stopped}}$;
\end{itemize}
Indeed, with these choices we have $\bar{h}_{\Hpub}(m)=H_{\mathrm{pub}}\cdot \delta_t(h_{\Hpub}^{\mathrm{(stopped)}}(m))=h_{\Hpub}(m)$. We also remark that in mCFS$_c$ there are other marginal differences with respect to our generalisation, which however do not impact on the main points of the scheme that we sketched above (e.g. computing $h_{\Hpub}(h_{\Hpub}(m)||R)$ instead of $h_{\Hpub}(m)$).

\begin{Th}
Let $(\sk,\pk)$ be the output of $\kgen_{\widetilde{\text{CFS}}}(1^\lambda)$. An attacker knowing the public key $\pk$ can forge a signature compatible to the private key $\sk$ for any message $m$.
\end{Th}

\begin{proof}
An attacker $\adv$ knowing the public parameters $(h,\gamma_t,t,\Hpub)$ is able to forge any signature. Instead of performing the steps of the signature algorithm, $\adv$ performs the following:
\begin{enumerate}
\item Given any message $m$, compute $x=\gamma_t(h(m))$;
\item output $x$ as the signature of $m$.
\end{enumerate}
Indeed, $x$ is a valid error vector in $\left(\mathbb{F}_2\right)^n$ of Hamming weight at most $t$ (by definition of $\gamma_t$) whose syndrome with respect to the parity-check matrix $\Hpub$ is $s=\bar{h}_{\Hpub}(m)$. Therefore, any verifier obtains
$$
\Hpub x^{\top}=\Hpub x^\top = \Hpub \cdot \gamma_t(h(m)) = \bar{h}_{\Hpub}(m)\;
$$
and the signature results valid.
\end{proof}
We remark how an attacker does not need to know the private key, and the number of operations performed by $\adv$ to successfully obtain a forgery are less than the number of operations performed by a honest user to obtain a valid signature. The key-point of the vulnerability of the scheme is that, to obtain a decodable syndrome, we force the application of a function $\gamma_t$ to the output of the hash function before computing the syndrome. Even though this step allows us to obtain a decodable syndrome without having to rely to the (expensive) re-iteration of the signature steps of the original CFS protocol, during the signature algorithm we are forced to explicitly determine a decodable error compatible with the output syndrome.

\section{Conclusions}
One of the practical issues of the CFS signature scheme is the computational effort required to obtain a decodable syndrome from the hash of the message. In \cite{ren2017efficient} the authors attempt to overcome this problem using a Merkel-Damgard-style code-based hash function from the space of binary strings into the set of decodable syndromes, significantly reducing the cost of signing. This approach has proven unsuccessful, since the protocol allows to an attacker who does not know the private key to produce a valid signature.

We showed that a generalization of this approach remains insecure: a hash function that sends arbitrarily long binary strings into the set of decodable syndromes can be constructed and yet there exists an attack on this new variation of the CFS signature. Therefore, other solutions should be found, in order to preserve the original security of the scheme but also to reduce the computational effort used in the signing process. The design of a suitable code-based signature should keep in mind both the provable security of CFS-like signatures and the efficiency of the KKS scheme.

\subsubsection*{Acknowledgements}
The first author acknowledges support from TIM S.p.A. through the PhD scholarship. The second author is a member of the INdAM Research group GNSAGA and of the Cryptography and Coding group of the Unione Matematica Italiana (UMI). The core of this work is contained in the third author's M.Sc. thesis.

\bibliographystyle{plain}
\bibliography{main}

\begin{thebibliography}{10}

\bibitem{augot2005family}
Daniel Augot, Matthieu Finiasz, and Nicolas Sendrier.
\newblock A family of fast syndrome based cryptographic hash functions.
\newblock In {\em International Conference on Cryptology in Malaysia}, pages
  64--83. Springer, 2005.

\bibitem{barreto2011one}
Paulo~SLM Barreto, Rafael Misoczki, and Marcos~A Simplicio~Jr.
\newblock One-time signature scheme from syndrome decoding over generic
  error-correcting codes.
\newblock {\em Journal of Systems and Software}, 84(2):198--204, 2011.

\bibitem{berlekamp1973goppa}
Elwyn Berlekamp.
\newblock Goppa codes.
\newblock {\em IEEE Transactions on Information Theory}, 19(5):590--592, 1973.

\bibitem{berlekamp1978inherent}
Elwyn Berlekamp, Robert McEliece, and Henk Van~Tilborg.
\newblock On the inherent intractability of certain coding problems (corresp.).
\newblock {\em IEEE Transactions on Information Theory}, 24(3):384--386, 1978.

\bibitem{eddsa}
Daniel~J. Bernstein, Niels Duif, Tanja Lange, Peter Schwabe, and Bo-Yin Yang.
\newblock High-speed high-security signatures.
\newblock {\em Journal of Cryptographic Engineering}, 2:77--89, 2012.

\bibitem{cayrel2007kabatianskii}
Pierre-Louis Cayrel, Ayoub Otmani, and Damien Vergnaud.
\newblock On kabatianskii-krouk-smeets signatures.
\newblock In {\em International Workshop on the Arithmetic of Finite Fields},
  pages 237--251. Springer, 2007.

\bibitem{DBLP:journals/corr/abs-2107-11082}
Nicola~Di Chiano, Riccardo Longo, Alessio Meneghetti, and Giordano Santilli.
\newblock A survey on {NIST} {PQ} signatures.
\newblock {\em CoRR}, abs/2107.11082, 2021.

\bibitem{chou2020classic}
Tung Chou, Carlos Cid, Simula UiB, Jan Gilcher, Tanja Lange, Varun Maram,
  Rafael Misoczki, Ruben Niederhagen, Kenneth~G Paterson, Edoardo Persichetti,
  et~al.
\newblock Classic mceliece: conservative code-based cryptography 10 october
  2020.
\newblock 2020.

\bibitem{courtois2001achieve}
Nicolas~T Courtois, Matthieu Finiasz, and Nicolas Sendrier.
\newblock How to achieve a mceliece-based digital signature scheme.
\newblock In {\em International Conference on the Theory and Application of
  Cryptology and Information Security}, pages 157--174. Springer, 2001.

\bibitem{dallot2007towards}
L{\'e}onard Dallot.
\newblock Towards a concrete security proof of courtois, finiasz and sendrier
  signature scheme.
\newblock In {\em Western European Workshop on Research in Cryptology}, pages
  65--77. Springer, 2007.

\bibitem{damgaard1989design}
Ivan~Bjerre Damg{\aa}rd.
\newblock A design principle for hash functions.
\newblock In {\em Conference on the Theory and Application of Cryptology},
  pages 416--427. Springer, 1989.

\bibitem{schnorr}
David~Mandell Freeman.
\newblock Schnorr identification and signatures.

\bibitem{gaborit2012efficient}
Philippe Gaborit and Julien Schrek.
\newblock Efficient code-based one-time signature from automorphism groups with
  syndrome compatibility.
\newblock In {\em 2012 IEEE International Symposium on Information Theory
  Proceedings}, pages 1982--1986. IEEE, 2012.

\bibitem{ecdsa}
Don Johnson, Alfred Menezes, and Scott Vanstone.
\newblock The elliptic curve digital signature algorithm ({ECDSA}).
\newblock {\em Int. J. Inf. Sec.}, 1:36--63, 08 2001.

\bibitem{kabatianskii1997digital}
Gregory Kabatianskii, Evgenii Krouk, and Ben Smeets.
\newblock A digital signature scheme based on random error-correcting codes.
\newblock In {\em IMA International Conference on Cryptography and Coding},
  pages 161--167. Springer, 1997.

\bibitem{kabatiansky2005error}
Grigorii Kabatiansky, Evgenii Krouk, and Sergei Semenov.
\newblock {\em Error correcting coding and security for data networks: analysis
  of the superchannel concept}.
\newblock John Wiley \& Sons, 2005.

\bibitem{kerry2013fips}
Cameron~F Kerry and Charles~Romine Director.
\newblock Fips pub 186-4 federal information processing standards publication
  digital signature standard (dss).
\newblock 2013.

\bibitem{kerry2013digital}
Cameron~F Kerry and Patrick~D Gallagher.
\newblock Digital signature standard (dss).
\newblock {\em FIPS PUB}, pages 186--4, 2013.

\bibitem{mceliece1978public}
Robert~J McEliece.
\newblock A public-key cryptosystem based on algebraic.
\newblock {\em Coding Thv}, 4244:114--116, 1978.

\bibitem{niederreiter1986knapsack}
Harald Niederreiter.
\newblock Knapsack-type cryptosystems and algebraic coding theory.
\newblock {\em Prob. Contr. Inform. Theory}, 15(2):157--166, 1986.

\bibitem{otmani2011efficient}
Ayoub Otmani and Jean-Pierre Tillich.
\newblock An efficient attack on all concrete kks proposals.
\newblock In {\em International Workshop on Post-Quantum Cryptography}, pages
  98--116. Springer, 2011.

\bibitem{ren2017efficient}
Fang Ren, Dong Zheng, WeiJing Wang, et~al.
\newblock An efficient code based digital signature algorithm.
\newblock {\em Int. J. Netw. Secur.}, 19(6):1072--1079, 2017.

\bibitem{shor1994algorithms}
Peter~W Shor.
\newblock Algorithms for quantum computation: discrete logarithms and
  factoring.
\newblock In {\em Proceedings 35th annual symposium on foundations of computer
  science}, pages 124--134. Ieee, 1994.

\end{thebibliography}
\end{document}